\documentclass[11pt,reqno]{amsart}

\usepackage{mathrsfs}
\usepackage{amsmath}
\usepackage{amssymb}
\usepackage{amsfonts}
\usepackage{amsthm}
\usepackage{a4wide}
\usepackage{graphicx}
\usepackage{color}
\usepackage[sans]{dsfont}
\usepackage{linearA}
\usepackage{pgfkeys}
\usepackage{longtable}
\usepackage{pdflscape}
\usepackage{float}
\usepackage{cancel}
\usepackage{dsfont}
\usepackage{setspace}
\usepackage{array}
\usepackage{hyperref}
\hypersetup{colorlinks=true,allcolors=[rgb]{0,0,0.6}}
\usepackage[english]{babel}
\usepackage[applemac]{inputenc}
\usepackage[T1]{fontenc}
\usepackage{tikz,pgflibraryshapes}
\usepackage{enumitem}
\usepackage{eufrak}
\usepackage[square,numbers,sort&compress]{natbib}
\usepackage[shadow,textwidth=22mm,textsize=tiny]{todonotes}
\usepackage[left=2.5cm,right=2.5cm,top=3cm,bottom=3cm]{geometry}

\newtheorem{theorem}{Theorem}[section]
\newtheorem{lemma}[theorem]{Lemma}
\newtheorem{proposition}[theorem]{Proposition}

\newtheorem{definition}[theorem]{Definition}

\newtheorem{assumption}{Hypothesis}

\numberwithin{equation}{section}
\DeclareMathOperator{\slim}{s-lim}

\author[J. Faupin]{J{\'e}r{\'e}my Faupin}
\address[J. Faupin]{Institut Elie Cartan de Lorraine \\
Universit{\'e} de Lorraine,
57045 Metz Cedex 1, France}
\email{jeremy.faupin@univ-lorraine.fr}

\begin{document}
\bibliographystyle{abbrv} \title[Generic asymptotic completeness]{Generic nature of asymptotic completeness in dissipative scattering theory}

\begin{abstract}
We review recent results obtained in the scattering theory of dissipative quantum systems representing the long-time evolution of a system $S$ interacting with another system $S'$ and susceptible of being absorbed by $S'$. The effective dynamics of $S$ is generated by an operator of the form $H = H_0 + V - \mathrm{i} C^* C$ on the Hilbert space of the pure states of $S$, where $H_0$ is the self-adjoint generator of the free dynamics of $S$, $V$ is symmetric and $C$ is bounded. The main example is a neutron interacting with a nucleus in the nuclear optical model. We recall the basic objects of the scattering theory for the pair $(H,H_0)$, as well as the results, proven in \cite{FaFr18_01,FaNi18_01}, on the spectral singularities of $H$ and the asymptotic completeness of the wave operators. Next, for the nuclear optical model, we show that asymptotic completeness generically holds.
\end{abstract}

\maketitle

\section{Introduction}

When a physical quantum system interacts with another one, part of its energy may be irreversibly transferred to the other system. This phenomenon of irreversible loss of energy is usually called \emph{quantum dissipation}. In particular, fundamentally, quantum systems cannot be completely isolated from their environment and, therefore, any quantum system experiences quantum dissipation to some extent, due to interactions with the environment.

This paper is concerned with the mathematical study of effective or empirical models of quantum dissipation. We consider a quantum system $S$ interacting with another quantum system $S'$. Our main concern is the understanding of the phenomenon of “capture”: We aim at studying models allowing for the description of both elastic scattering and absorption of $S$ by $S'$. Such models apply to various physical situations, especially to neutrons interacting with nuclei in the \emph{nuclear optical model} (see Section \ref{sec:nuclear}).

In \cite{FaFr18_01,FaNi18_01}, the scattering theory for a class of abstract \emph{pseudo-Hamiltonians} on a Hilbert space $\mathcal{H}$ is studied. In the abstract setting considered in \cite{FaFr18_01,FaNi18_01}, the pseudo-Hamiltonian corresponding to the generator of the effective dynamics of the system $S$ is given by
\begin{equation*}
H = H_0 + V - \mathrm{i} C^*C,
\end{equation*}
where $H_0$ is a self-adjoint operator on $\mathcal{H}$ with purely absolutely continuous spectrum, $V$ is symmetric and relatively compact with respect to $H_0$, and $C$ is bounded and relatively compact with respect to $H_0$. The operator $H_0$ is  the generator of the unitary free dynamics of $S$ while $V-\mathrm{i}C^*C$ represents the effective interaction between $S$ and $S'$. The main purpose in \cite{FaFr18_01,FaNi18_01} is then to study the scattering theory for the pair $(H,H_0)$. Suitable hypotheses on $H_0$, $V$ and $C$ are formulated in such a way that they can be verified in the particular case where $H$ is given by a dissipative Schrödinger operator. See the next sections for more details.

Prior to \cite{FaFr18_01,FaNi18_01}, mathematical scattering theory for dissipative operators on Hilbert spaces has been considered by many authors (see, e.g., \cite{Ka66_01,Ma75_01,Da78_01,Da80_01,FaFaFrSc17_01,WaZu14_01} and references therein). In these references, in particular, the existence of the wave operators associated to $H$ and $H_0$ is established under various conditions. In \cite{FaFr18_01,FaNi18_01}, the \emph{asymptotic completeness} of the wave operators is studied. It is shown that, under suitable assumptions, asymptotic completeness is equivalent to the absence of \emph{spectral singularities} embedded into the essential spectrum of $H$.

Our purpose here is twofold. First, we review the results established in \cite{FaFr18_01,FaNi18_01}. Next, for the nuclear optical model, we prove that generically (in a Baire category sense), the pseudo-Hamiltonian $H$ has no spectral singularities embedded in its essential spectrum. This implies that the wave operators are generically asymptotically complete.

The paper is organized as follows. In Section \ref{section:2}, we introduce the main objects involved in dissipative scattering theory and we recall their basic properties. Section \ref{section:3} concerns the notions of spectral singularities and asymptotic completeness, as well as the results proven in \cite{FaFr18_01,FaNi18_01}. Finally, in Section \ref{section:4}, we state and prove our new result on the generic nature of asymptotic completeness.


\section{Mathematical setting}\label{section:2}
As mentioned in the introduction, we consider a quantum system $S$ interacting with another quantum system $S'$ and susceptible of being absorbed by $S'$. The pure states of $S$ correspond to the normalized vectors in a complex Hilbert space $\mathcal{H}$. The scalar product in $\mathcal{H}$ is denoted by $\langle \cdot , \cdot \rangle$. The effective dynamics of $S$ is supposed to be generated by a pseudo-Hamiltonian acting on $\mathcal{H}$, of the form
\begin{equation*}
H = H_0 + V - \mathrm{i} C^*C = H_V - \mathrm{i} C^*C,
\end{equation*}
where $H_0$ is a self-adjoint operator on $\mathcal{H}$ corresponding to the generator of the free dynamics of $S$ and $V - \mathrm{i} C^*C$ is an effective interaction term due to the presence of $S'$.

In this section, we state the abstract assumptions on the operators $H_0$, $V$ and $C$ which were introduced in \cite{FaFr18_01,FaNi18_01} in order to establish results on the spectral and scattering theories for the pair $(H,H_0)$. In the next section, we will recall that those abstract assumptions are fulfilled in our main example, namely the nuclear optical model. In this model, $H$ is a dissipative Schrödinger operator, with $H_0 = - \Delta$ on $L^2( \mathbb{R}^3 )$ and $V$, $C$ multiplication operators by bounded, real-valued potentials decaying sufficiently fast at $\infty$ (see Section \ref{sec:nuclear} for more details).

To shorten notations below, the resolvents of the operators $H_0$, $H_V$ and $H$ are denoted by
\begin{equation*}
R_0(z) = ( H_0 - z )^{-1} , \quad R_V(z) = ( H_V - z )^{-1}, \quad R(z) = ( H - z )^{-1} ,
\end{equation*}
for any $z$ in the resolvent set of the corresponding operator.

\subsection{Basic assumptions}

The set of bounded operators on $\mathcal{H}$ is denoted by $\mathcal{L}( \mathcal{H} )$. We recall that an operator $B$ is called relatively compact with respect to a self-adjoint operator $A$ if $\mathcal{D}( A ) \subset \mathcal{D}( B )$ and $B ( A + \mathrm{i} )^{-1}$ is compact. The following basic assumptions are made:
\begin{assumption}[Basic assumptions]\label{hyp:H0}$ $
\begin{enumerate}[label=\rm{(\roman*)},leftmargin=0.8cm,itemsep=1pt]
\item $H_0 \ge 0$ (or, more generally, $H_0$ is self-adjoint and semi-bounded from below),
\item $V$ is symmetric and relatively compact with respect to $H_0$,
\item $C \in \mathcal{L}( \mathcal{H} )$ and $C$ is relatively compact with respect to $H_0$.
\end{enumerate}
\end{assumption}
We recall that an operator $A$ on $\mathcal{H}$ is called dissipative if, for all $u \in \mathcal{D}( A )$, $\mathrm{Im}( \langle u , A u \rangle ) \le 0$. Moreover, $A$ is called maximal dissipative if $A$ is dissipative and has no proper dissipative extension. Hypothesis \ref{hyp:H0} has the following simple consequences.
\begin{proposition}
Suppose that Hypothesis \ref{hyp:H0} holds. Then
\begin{enumerate}[label=\rm{(\arabic*)},leftmargin=0.8cm,itemsep=1pt]
\item $H_V = H_0 + V$ is a self-adjoint operator on $\mathcal{H}$ with domain
$
\mathcal{D}( H_V ) = \mathcal{D}( H_0 ) .
$
\item $H = H_V - \mathrm{i} C^*C$ is a maximal dissipative operator on $\mathcal{H}$ with domain
$
\mathcal{D}( H ) = \mathcal{D}( H_0 ) .
$
\item The operator $-\mathrm{i} H$ generates a strongly continuous group $\{ e^{-\mathrm{i}tH} \}_{ t \in \mathbb{R} }$ such that
\begin{equation*}
\big \|e^{-\mathrm{i}tH}\big\| \le 1 \, \text{ if } \,  t \ge 0 , \qquad \big\|e^{-\mathrm{i}tH}\big\| \le e^{\|C^*C\||t|} \, \text{ if } \, t \le 0 .
\end{equation*}
In particular, $-\mathrm{i} H$ generates the strongly continuous semigroup of contractions $\{ e^{-\mathrm{i}tH} \}_{ t \ge 0 }$.
\item The adjoint of $H$ is
\begin{equation*}
H^* = H_0 + V + \mathrm{i} C^* C ,
\end{equation*}
with domain $\mathcal{D}( H^* ) = \mathcal{D}( H_0 )$. Moreover, $\mathrm{i} H^*$ generates of a strongly continuous group $\{ e^{\mathrm{i}tH^*} \}_{t \in \mathbb{R}}$ such that $\{ e^{\mathrm{i}tH^*} \}_{t \ge 0}$ is a semigroup of contractions.
\end{enumerate}
\end{proposition}
\begin{proof}
For the convenience of the reader, we sketch some of the arguments which were eluded in \cite{FaFr18_01,FaNi18_01}.

(1) is a simple consequence of the Kato-Rellich Theorem together with the fact that $V$ is symmetric and relatively compact with respect to $H_0$, and hence infinitesimally small with respect to $H_0$ (see, e.g., \cite[Corollary 2, p. 113]{ReSi80_01}).

To prove (2), one observes that $H$ is dissipative since, for all $u \in \mathcal{D}( H ) = \mathcal{D}( H_0 )$,
\begin{equation*}
\mathrm{Im}( \langle u , H u \rangle ) = - \| C u \|^2 \le 0.
\end{equation*}
To verify that $H$ is maximal dissipative, by a theorem of Phillips \cite{Ph59_01}, it then suffices to show that $\mathrm{Ran}( H - \mathrm{i} \lambda ) = \mathcal{H}$ for some $\lambda > 0$. This easily follows from the fact that $H - \mathrm{i} \lambda : \mathcal{D}( H_0 ) \to \mathcal{H}$ is invertible for $\lambda > \| C^* C \|$ (here one uses that $H_V$ is self-adjoint, and hence that $\| ( H_V - \mathrm{i} \lambda )^{-1} \| \le \lambda^{-1}$).

(3) Since $H_V$ is self-adjoint, $-\mathrm{i}H_V$ generates a strongly continuous unitary group $\{ e^{-\mathrm{i}tH_V} \}_{ t \in \mathbb{R} }$. Hence, since $C^*C$ is bounded, a perturbation argument (see, e.g., \cite[Theorem 11.4.1]{Da07_01}) shows that $-\mathrm{i}H$ generates a strongly continuous group $\{ e^{-\mathrm{i}tH} \}_{ t \in \mathbb{R} }$ such that $\|e^{-\mathrm{i}tH} \| \le e^{\|C^*C\||t|}$ for all $t \in \mathbb{R}$. The fact that $e^{-\mathrm{i}tH}$ is a contraction for $t \ge 0$ is a consequence of the fact that $H$ is maximal dissipative (see e.g. \cite[Theorem 10.4.2]{Da07_01}).

(4) Standard arguments show that the adjoint of $H$ is given by $H^* = H_0 + V + \mathrm{i} C^*C$ with domain $\mathcal{D}( H^* ) = \mathcal{D}( H_0 )$. One then verifies, in the same way as for $-\mathrm{i} H$, that $\mathrm{i} H^*$ generates of a strongly continuous group $\{ e^{\mathrm{i}tH^*} \}_{t \in \mathbb{R}}$ such that $\{ e^{\mathrm{i}tH^*} \}_{t \ge 0}$ is a semigroup of contractions
\end{proof}

The contraction semigroup $\{ e^{-\mathrm{i}tH} \}_{ t \ge 0 }$ has the interpretation of a dynamics in the following sense. If $u_0 \in \mathcal{H}$, $\| u_0 \|=1$, represents the initial state of the quantum system $S$ at time $t=0$, then the state of $S$ at a positive time $t>0$ is given by $\| u_t \|^{-1} u_t$, with $u_t := e^{-\mathrm{i}tH} u_0$. Here it should be noted that $\| u_t \| \le 1$ for all $t \ge 0$ since $e^{-\mathrm{i}tH}$ is a contraction, and that $u_t \neq 0$ since $ e^{-\mathrm{i}tH} $ is invertible.

\subsection{Spectrum and spectral subspaces of $\mathcal{H}$}\label{subsec:spectrum}
Since $H$ is maximal dissipative -- or equivalently $-\mathrm{i}H$ generates a strongly continuous semigroup of contractions -- an application of the Hille-Yosida Theorem shows that the spectrum of $H$ satisfies
\begin{equation*}
\sigma ( H ) \subset \{ z \in \mathbb{C} , \, \mathrm{Im}( z ) \le 0 \}.
\end{equation*}
In this section, we review the definitions of some spectral subspaces of $H$.

\subsubsection{The space of bound states}
If $\mathcal{D}$ is a subset of $\mathcal{H}$, we denote by $\overline{\mathcal{D}}$ its closure.
\begin{definition}[Space of bound states $\mathcal{H}_{\mathrm{b}}(H)$]
Suppose that Hypothesis \ref{hyp:H0} holds. The space of bound states of $H$ is defined as the closure of the vector space spanned by all eigenvectors of $H$ corresponding to real eigenvalues, \textit{i.e.}
\begin{equation*}
\mathcal{H}_{\mathrm{b}}( H ) := \overline{\mathrm{Span} \{ u \in \mathcal{D}( H ) , \, \exists \lambda \in \mathbb{R} , \, Hu = \lambda u \} }.
\end{equation*}
Similarly, 
\begin{equation*}
\mathcal{H}_{\mathrm{b}}( H^* ) := \overline{\mathrm{Span} \{ u \in \mathcal{D}( H^* ) , \, \exists \lambda \in \mathbb{R} , \, H^*u = \lambda u \} }.
\end{equation*}
\end{definition}
In the particular case were $H$ is self-adjoint, \textit{i.e.} $C=0$, we see that the space of bound states identifies with the pure point spectral subspace of $H$ usually denoted by $\mathcal{H}_{ \mathrm{pp} }( H )$. In general, $\mathcal{H}_{\mathrm{b}}( H )$ and the pure point spectral subspace of $H_V$ are related as follows.
\begin{proposition}\label{prop:Hb}
Suppose that Hypothesis \ref{hyp:H0} holds. Then
\begin{equation*}
\mathcal{H}_{\mathrm{b}}( H ) = \mathcal{H}_{\mathrm{b}}( H^* ) \subset \mathcal{H}_{\mathrm{pp}}( H_V ) \cap \mathrm{Ker}( C ).
\end{equation*}
\end{proposition}
\begin{proof}
See \cite[Lemma 3.1]{FaFr18_01}.
\end{proof}

\subsubsection{Discrete and essential spectra}
The discrete and essential spectra of $H$ may be defined as follows. We recall that an operator $A$ on $\mathcal{H}$ with domain $\mathcal{D}(A)$ is called Fredholm if $\mathrm{Ran}( A - \lambda \mathrm{Id} )$ is closed, $\mathrm{dim} \, \mathrm{Ker}( A - \lambda \mathrm{Id} ) < \infty$ and $\mathrm{codim} \, \mathrm{Ran}( A - \lambda \mathrm{Id} ) < \infty$.
\begin{definition}[Discrete spectrum]
Suppose that Hypothesis \ref{hyp:H0} holds. The discrete spectrum of $H$, denoted by $\sigma_{\mathrm{disc}}(H)$, is the set of isolated eigenvalues of $H$ with finite algebraic multiplicity. In other words, $\lambda \in \sigma_{\mathrm{disc}}(H)$ if $\lambda$ is an isolated point in $\sigma( H )$, there exists $u \in \mathcal{D}( H ) \setminus \{ 0 \}$ such that $H u = \lambda u$ and $\mathrm{dim} \, \mathrm{Ker}( H - \lambda \mathrm{Id} ) < \infty$.
\end{definition}
\begin{definition}[Essential spectrum]
Suppose that Hypothesis \ref{hyp:H0} holds. The essential spectrum of $H$, denoted by $\sigma_{\mathrm{ess}}(H)$, is the set of $\lambda \in \mathbb{C}$ such that $H - \lambda \mathrm{Id}$ is not Fredholm.
\end{definition}
We mention that other possible definitions of the essential spectrum for non self-adjoint operators may be found in the literature (see, e.g., \cite[Section IX]{EdEv87_01}) but these different definitions coincide in our context \cite[Theorem IX.1.6]{EdEv87_01}. The discrete and essential spectra of $H$ are related as follows.
\begin{proposition}
Suppose that Hypothesis \ref{hyp:H0} holds. Then
\begin{equation*}
\sigma_{\mathrm{ess}}(H_0) = \sigma_{\mathrm{ess}}(H_V) = \sigma_{\mathrm{ess}}( H ) = \sigma_{\mathrm{ess}}( H^* ) = \sigma(H) \setminus \sigma_{\mathrm{disc}}(H) = \sigma(H^*) \setminus \sigma_{\mathrm{disc}}(H^*).
\end{equation*}
\end{proposition}
\begin{proof}
The first two equalities are consequences of the facts that $V$ and $C^*C$ are relatively compact perturbations of $H_0$ (see e.g. \cite[Theorem 11.2.6]{Da07_01}). The last equality is proven e.g. in \cite[Theorem IX.1.6]{EdEv87_01}.
\end{proof}

Summing up, the spectrum of $H$ is of the form pictured in Figure \ref{fig1}.
\begin{figure}[H] 
\begin{center}
\begin{tikzpicture}[scale=0.6, every node/.style={scale=0.9}]

   \draw[->](-4,2) -- (12,2);      
      \draw[->](4,-2) -- (4,4);      
  \draw[-, very thick] (4,2)--(11.95,2);

    \draw[-] (-1.9,1.9)--(-2.1,2.1);
    \draw[-] (-1.9,2.1)--(-2.1,1.9);

    \draw[-] (0.9,1.9)--(1.1,2.1);
    \draw[-] (0.9,2.1)--(1.1,1.9);

    \draw[-] (3.1,1.9)--(3.3,2.1);
    \draw[-] (3.1,2.1)--(3.3,1.9);

    \draw[-] (6.45,0.2)--(6.65,0.4);
    \draw[-] (6.45,0.4)--(6.65,0.2);      

    \draw[-] (6.65,0.7)--(6.85,0.9);
    \draw[-] (6.65,0.9)--(6.85,0.7);   
   
    \draw[-] (6.8,1.1)--(7,1.3);
    \draw[-] (6.8,1.3)--(7,1.1);   
       
    \draw[-] (6.9,1.4)--(7.1,1.6);
    \draw[-] (6.9,1.6)--(7.1,1.4);   
   
    \draw[-] (6.95,1.6)--(7.15,1.8);
    \draw[-] (6.95,1.8)--(7.15,1.6);   
   
    \draw[-] (6.97,1.8)--(7.17,2);
    \draw[-] (6.97,2)--(7.17,1.8);

    \draw[-] (10.6,1.6)--(10.8,1.8);
    \draw[-] (10.6,1.8)--(10.8,1.6);

    \draw[-] (9.5,-0.9)--(9.7,-0.7);
    \draw[-] (9.5,-0.7)--(9.7,-0.9);

    \draw[-] (8.5,0)--(8.7,0.2);
    \draw[-] (8.5,0.2)--(8.7,0);

    \draw[-] (4.1,0.9)--(4.3,0.7);
    \draw[-] (4.1,0.7)--(4.3,0.9);

    \draw[-] (-0.5,1.2)--(-0.7,1.4);
    \draw[-] (-0.5,1.4)--(-0.7,1.2);

    \draw[-] (1.5,-0.2)--(1.7,-0.4);
    \draw[-] (1.5,-0.4)--(1.7,-0.2);

        \end{tikzpicture}
\caption{ \footnotesize  \textbf{Spectrum of $H$.} The spectrum of $H$ is contained in the lower half-plane. The essential spectrum of $H$ coincides with that of $H_0$ and is contained in $[ 0 , \infty )$. The discrete spectrum of $H$ consists of isolated eigenvalues of finite algebraic multiplicities which may accumulate at any point of the essential spectrum. }\label{fig1}
\end{center}
\end{figure}
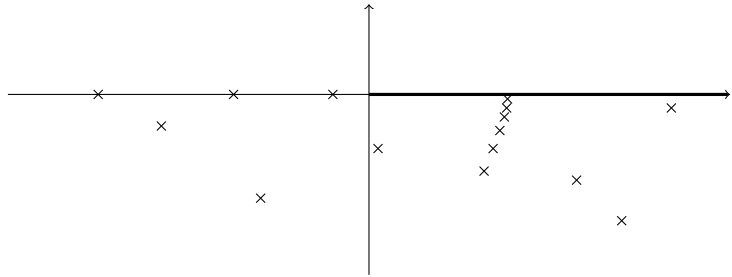

For $\lambda \in \sigma_{\mathrm{disc}}(H)$, the Riesz projection corresponding to $\lambda$, denoted by $\pi_\lambda$, is defined by
\begin{equation*}
\pi_\lambda := \frac{1}{2\mathrm{i}\pi} \int_\gamma ( z \mathrm{Id} - H )^{-1} \mathrm{d} z,
\end{equation*}
where $\gamma$ is a circle centered at $\lambda$, oriented counterclockwise and such that $\lambda$ is the only point of $\sigma(H)$ contained in the interior of $\gamma$. We recall that a vector $u \in \mathcal{H} $ is called a generalized eigenvector corresponding to $\lambda$ if there exists a positive integer $k$ such that $u \in \mathcal{D}( H^k )$ and $(H-\lambda)^k u = 0$. As is well-known, for $\lambda \in \sigma_{\mathrm{disc}}( H )$, the range of the Riesz projection $\pi_\lambda$ coincides with the vector space spanned by all generalized eigenvectors corresponding to $\lambda$.
\begin{proposition}\label{prop:4}
Suppose that Hypothesis \ref{hyp:H0} holds and let $\lambda \in \sigma_{\mathrm{disc}}(H)$. Then $\pi_\lambda$ is a projection such that $\mathrm{dim} \, \mathrm{Ran}( \pi_\lambda ) < \infty$ and
\begin{equation*}
\mathrm{Ran}( \pi_\lambda ) = \big \{ u \in \mathcal{D}( H^k ) , \, ( H - \lambda )^k u = 0 , \, \text{ for some } k \in \mathbb{N}, \, 1 \le k \le \mathrm{dim} \, \mathrm{Ran}( \pi_\lambda ) \big \}.
\end{equation*}
\end{proposition}
\begin{proof}
See, e.g., \cite[Theorem 1.5.4]{Da07_01}.
\end{proof}
In the particular case where $\lambda$ is a \emph{real} isolated eigenvalue of $H$, one can prove that the only possible generalized eigenvectors corresponding to $\lambda$ are eigenvectors in the usual sense.
\begin{proposition}\label{prop:5}
Suppose that Hypothesis \ref{hyp:H0} holds and let $\lambda \in \sigma_{\mathrm{disc}}(H) \cap \mathbb{R}$. Then
\begin{equation*}
\mathrm{Ran}( \pi_\lambda ) = \{ u \in \mathcal{D}( H ) , \, ( H - \lambda ) u = 0 \}.
\end{equation*}
\end{proposition}
\begin{proof}
See \cite[Lemma 3.3]{FaFr18_01}.
\end{proof}
Of course, one can define Riesz projections in the same way for $H^*$ and verify that statements analogous to Propositions \ref{prop:4}--\ref{prop:5} hold for $H^*$.

\subsubsection{The dissipative space}
\begin{definition}[Space $\mathcal{H}_{\mathrm{d}}(H)$]
Suppose that Hypothesis \ref{hyp:H0} holds. The dissipative space, or space of decaying states of $H$, is defined by
\begin{equation*}
\mathcal{H}_{ \mathrm{d} }( H ) := \big \{ u \in \mathcal{H} , \, \lim_{ t \to \infty } \big \| e^{ - \mathrm{i} t H } u \big \| = 0 \big \}.
\end{equation*}
Likewise, 
\begin{equation*}
\mathcal{H}_{ \mathrm{d} }( H^* ) := \big \{ u \in \mathcal{H} , \, \lim_{ t \to \infty } \big \| e^{ \mathrm{i} t H^* } u \big \| = 0 \big \}.
\end{equation*}
\end{definition}
Since $\{ e^{ - \mathrm{i} t H } \}_{t \ge 0 }$ and $\{ e^{ \mathrm{i} t H^* } \}_{ t \ge 0 }$ are contraction semigroups, it is easy to verify that $\mathcal{H}_{\mathrm{d}}( H )$ and $\mathcal{H}_{\mathrm{d}}( H^* )$ are closed. Moreover, it should be observed that the semigroup property implies that, for all $u \in \mathcal{H}$, the map $[ 0 , \infty ) \ni t \mapsto \| e^{-\mathrm{i}tH} u \|$ is decreasing and hence the limit $\lim_{ t \to \infty } \| e^{ - \mathrm{i} t H } u \|$ exists for all $u \in \mathcal{H}$. One can actually define the probabilities of elastic scattering and absorption as follows. Let $u_0 \in \mathcal{H}_{\mathrm{b}}(H)^\perp$, $\| u_0 \| = 1$, be an initial state orthogonal to all bound states of $H$. The probability of elastic scattering of the system $S$, initially in the state $u_0$, is defined by
\begin{equation*}
p_{\mathrm{scatt}}(u_0) := \lim_{ t \to \infty } \big \| e^{-\mathrm{i}tH} u_0 \big \|^2.
\end{equation*}
Likewise, the probability of absorption of the system $S$, initially in the state $u_0$, is
\begin{equation*}
p_{\mathrm{abs}}(u_0) := 1 - \lim_{ t \to \infty } \big \| e^{-\mathrm{i}tH} u_0 \big \|^2.
\end{equation*}

We introduce the following definition.
\begin{definition}[Space $\mathcal{H}_{\mathrm{p}}(H)$]
Suppose that Hypothesis \ref{hyp:H0} holds. The subspace $\mathcal{H}_{\mathrm{p}}(H)$ is the closure of the vector space spanned by all generalized eigenvectors of $H$ corresponding to an eigenvalue with a strictly negative imaginary part,
\begin{equation*}
\mathcal{H}_{ \mathrm{p} }( H ) := \overline{ \big \{ u \in \mathrm{Ran}( \pi_\lambda ) , \, \lambda \in \sigma_{\mathrm{disc}}( H ) , \, \mathrm{Im}( \lambda ) < 0 \big \} }.
\end{equation*}
Likewise,
\begin{equation*}
\mathcal{H}_{ \mathrm{p} }( H^* ) := \overline{ \big \{ u \in \mathrm{Ran}( \pi_\lambda ) , \, \lambda \in \sigma_{\mathrm{disc}}( H^* ) , \, \mathrm{Im}( \lambda ) > 0 \big \} }.
\end{equation*}
\end{definition}
The following easy proposition shows that the dissipative space contains $\mathcal{H}_{\mathrm{p}}(H)$. We give the proof for the convenience of the reader.
\begin{proposition}\label{prop:Hp}
Suppose that Hypothesis \ref{hyp:H0} holds. Then
\begin{equation*}
\mathcal{H}_{ \mathrm{p} }( H ) \subseteq \mathcal{H}_{ \mathrm{d} }( H ) \subseteq \mathcal{H}_{ \mathrm{b} }( H )^\perp, \quad \mathcal{H}_{ \mathrm{p} }( H^* ) \subseteq \mathcal{H}_{ \mathrm{d} }( H^* ) \subseteq \mathcal{H}_{ \mathrm{b} }( H )^\perp.
\end{equation*}
\end{proposition}
\begin{proof}
First, we prove that $\mathcal{H}_{ \mathrm{p} }( H ) \subseteq \mathcal{H}_{ \mathrm{d} }( H )$. Let $\lambda \in \sigma_{\mathrm{disc}}( H )$, $\mathrm{Im}( \lambda ) < 0$ and let $u \in \mathrm{Ran}( \pi_\lambda )$. Let $k= \mathrm{dim} \, \mathrm{Ran}( \pi_\lambda ) < \infty$. We compute
\begin{align*}
\big \| e^{ - \mathrm{i} t H } u \big \| = e^{ t \mathrm{Im}( \lambda ) } \big \| e^{ - \mathrm{i} t ( H - \lambda ) } u \big \| = e^{ t \mathrm{Im}( \lambda ) } \Big \| \sum_{ j = 0 }^{ k - 1 } \frac{ ( - \mathrm{i} t )^j }{ j! } ( H - \lambda )^j u \Big \| \to 0 , \quad t \to \infty ,
\end{align*}
since $\mathrm{Im}( \lambda ) < 0$. Hence $u \in\mathcal{H}_{ \mathrm{d} }( H )$.

Next, we prove that $\mathcal{H}_{ \mathrm{d} }( H ) \subseteq \mathcal{H}_{ \mathrm{b} }( H )^\perp$. Let $u \in \mathcal{H}_{ \mathrm{d} }( H )$ and let $v$ be an eigenvector of $H^*$ corresponding to a real eigenvalue. We have that
\begin{align*}
\big | \langle v , u \rangle \big | = \big |  \langle  e^{ \mathrm{i} t H^* } v , e^{ - \mathrm{i} t H } u \rangle \big | = \big | \langle v , e^{ - \mathrm{i} t H } u \rangle \big | \le \| v \| \big \| e^{ - \mathrm{i} t H } u \big \| \to 0 , \quad t \to \infty.
\end{align*}
Hence $u$ is orthogonal to all eigenvectors of $H^*$ corresponding to real eigenvalues, and therefore $u \in \mathcal{H}_{ \mathrm{b} }( H^* )^\perp$. Since $\mathcal{H}_{ \mathrm{b} }( H ) = \mathcal{H}_{ \mathrm{b} }( H^* )$ by Proposition \ref{prop:Hb}, this concludes the proof.

The proof of $\mathcal{H}_{ \mathrm{p} }( H^* ) \subseteq \mathcal{H}_{ \mathrm{d} }( H^* )$ and $\mathcal{H}_{ \mathrm{d} }( H^* ) \subseteq \mathcal{H}_{ \mathrm{b} }( H )^\perp$ are analogous.
\end{proof}

\subsubsection{The absolutely continuous spectral subspace}

Now, we turn to a possible definition of an \emph{absolutely continuous spectral subspace} for the non-self-adjoint operator $H$, following Davies \cite{Da80_01}.
\begin{definition}[Absolutely continuous spectral subspace]\label{def:Hac}
Suppose that Hypothesis \ref{hyp:H0} holds. The absolutely continuous spectral subspace of $H$ is defined by
\begin{equation*}
\mathcal{H}_{ \mathrm{ac} }( H ) := \overline{ M(H) },
\end{equation*}
where
\begin{equation*}
M( H ) := \Big \{ u \in \mathcal{H} , \, \exists \mathrm{c}_u > 0 , \, \forall v \in \mathcal{H} , \, \int_0^\infty \big | \langle e^{ - \mathrm{i} t H } u , v \rangle \big |^2 \mathrm{d} t \le \mathrm{c}_u \| v \|^2 \Big \}.
\end{equation*}
The absolutely continuous spectral subspace of $H^*$ is defined similarly, replacing $e^{ - \mathrm{i} t H }$ by $e^{ \mathrm{i} t H^* }$ in the definition above.
\end{definition}
In the particular case where $H$ is self-adjoint, the definition of $\mathcal{H}_{ \mathrm{ac} }( H )$ coincides with the usual one based on the nature of the spectral measures of $H$. Moreover, if $H$ is self-adjoint, $M(H)$ is closed and hence $\mathcal{H}_{ \mathrm{ac} }( H ) = M( H )$. Another possible definition of an absolutely continuous spectral subspace of $H$ follows from the theory of unitary dilations of non-self-adjoint operators, see e.g., \cite{NaFoBeKe10_01}. The relevance of Definition \ref{def:Hac} may be supported by the following result.
\begin{proposition}
Suppose that Hypothesis \ref{hyp:H0} holds. Then
\begin{equation*}
\mathcal{H}_{\mathrm{ac}}( H ) = \mathcal{H}_{ \mathrm{b} }( H )^\perp.
\end{equation*}
In particular,
\begin{equation*}
\mathcal{H}_{ \mathrm{d} }( H ) \subseteq \mathcal{H}_{\mathrm{ac}}( H ) = \mathcal{H}_{\mathrm{ac}}( H^* ).
\end{equation*}
\end{proposition}
\begin{proof}
The fact that $\mathcal{H}_{\mathrm{ac}}( H ) = \mathcal{H}_{ \mathrm{b} }( H )^\perp$ is proven in \cite{Da80_01}. The second equation is a direct consequence of Propositions \ref{prop:Hb} and \ref{prop:Hp}.
\end{proof}
We mention that another natural -- and relevant -- definition for the absolutely continuous spectral subspace of $H$ would be the orthogonal complement of all generalized eigenstates of $H^*$, namely
\begin{equation*}
\tilde{\mathcal{H}}_{\mathrm{ac}}( H ) := \big ( \mathcal{H}_{ \mathrm{b} }( H ) \oplus \mathcal{H}_{ \mathrm{p} }( H^* ) \big )^\perp.
\end{equation*}
According to the previous proposition and Definition \ref{def:Hac}, we then have that
\begin{equation*}
\tilde{\mathcal{H}}_{\mathrm{ac}}( H ) := \overline{ \tilde M( H ) } , \quad \tilde M( H ) := \Big \{ u \in \mathcal{H}_{ \mathrm{p} }( H^* )^\perp , \, \exists \mathrm{c}_u > 0 , \, \forall v \in \mathcal{H} , \, \int_0^\infty \big | \langle e^{ - \mathrm{i} t H } u , v \rangle \big |^2 \mathrm{d} t \le \mathrm{c}_u \| v \|^2 \Big \}.
\end{equation*}

\subsection{The wave and scattering operators}

In this section we define the central objects in the scattering theory for the pair $(H , H_0)$, namely the wave operators, the scattering operator and the scattering matrices. We begin by introducing hypotheses insuring that these objects are well-defined.

\subsubsection{Hypotheses}

Recall that, given a self-adjoint operator $A$ on $\mathcal{H}$, $\mathcal{H}_{\mathrm{pp}}(A)$, $\mathcal{H}_{\mathrm{ac}}(A)$ and $\mathcal{H}_{\mathrm{sc}}(A)$ denote the pure point, absolutely continuous and singular continuous spectral subspaces of $A$, respectively. Likewise, $\sigma_{\mathrm{pp}}(A)$, $\sigma_{\mathrm{ac}}(A)$ and $\sigma_{\mathrm{sc}}(A)$ denote the pure point, absolutely continuous and singular continuous spectra of $A$.

The next hypothesis concerns the spectra of the self-adjoint operators $H_0$ and $H_V$ (recall that $H_0$ and $H_V$ have the same essential spectrum, assuming Hypothesis \ref{hyp:H0}).
\begin{assumption}[Spectra of $H_0$ and $H_V$]\label{hyp:H1}$ $
\begin{enumerate}[label=\rm{(\roman*)},leftmargin=0.8cm,itemsep=1pt]
\item The spectrum of $H_0$ is purely absolutely continuous, \textit{i.e.}, $\sigma_{\mathrm{ac}}(H_0)=\sigma(H_0)$, $\sigma_{\mathrm{pp}}(H_0)=\emptyset$, $\sigma_{\mathrm{sc}}(H_0)=\emptyset$.
\item $H_V$ has no singular spectrum, no embedded eigenvalues, and only finitely many eigenvalues counting multiplicity, \textit{i.e.}, $\sigma_{\mathrm{sc}}(H_V)=\emptyset$, $\sigma_{\mathrm{pp}}(H_V) \subset \mathbb{R} \setminus \sigma( H_0 )$ and $\mathrm{dim} \, \mathcal{H}_{ \mathrm{pp} }( H_V ) < \infty$.
\end{enumerate}
\end{assumption}
We denote by $\Pi_{\mathrm{ac}}(H_V)$ the orthogonal projection onto $\mathcal{H}_{\mathrm{ac}}(H_V)$. The symbol $\slim$ stands for strong limit. Our second hypothesis concerns the unitary wave operators associated to the self-adjoint pair $(H_V,H_0)$ (in the statement of Hypothesis \ref{hyp:H2} below, it is tacitly assumed that Hypothesis \ref{hyp:H1} holds).
\begin{assumption}[Wave operators for $(H_0,H_V)$]\label{hyp:H2}
The wave operators
\begin{equation*}
W_\pm(H_V,H_0) := \underset{t\to\pm\infty}{\slim} \, e^{\mathrm{i}tH_V}e^{-\mathrm{i}tH_0} , \quad W_\pm(H_0,H_V) := \underset{t\to\pm\infty}{\slim} \, e^{\mathrm{i}tH_0}e^{-\mathrm{i}tH_V} \Pi_{\mathrm{ac}}(H_V) ,
\end{equation*}
exist and are asymptotically complete, \textit{i.e.},
\begin{equation*}
\mathrm{Ran} \, W_\pm(H_V,H_0) = \mathcal{H}_{\mathrm{ac}}(H_V) = \mathcal{H}_{\mathrm{pp}}(H_V)^\perp , \quad \mathrm{Ran} \, W_\pm(H_0,H_V) = \mathcal{H} .
\end{equation*}
\end{assumption}
In our next assumption, we require that the operator $C$ be relatively smooth with respect to $H_V$ in the sense of Kato \cite{Ka66_01}.
\begin{assumption}[Relative smoothness of $C$ with respect to $H_V$]\label{hyp:H3}
There exists $\mathrm{c}_V>0$ such that, for all $u \in \mathcal{H}_{\mathrm{ac}}(H)$,
\begin{equation*}
\int_{\mathbb{R}} \big \|C e^{ - \mathrm{i} t H_V } u \big \|^2 \mathrm{d}t \le \mathrm{c}_V \| u \|^2.
\end{equation*}
\end{assumption}
In the remainder of this section, we recall properties of the wave and scattering operators for the pair $(H,H_0)$, assuming that Hypotheses \ref{hyp:H0}--\ref{hyp:H3} hold.

\subsubsection{The wave operators $W_-( H , H_0 )$ and $W_+( H^* , H_0 )$}

Assuming that $H_0$ has purely absolutely continuous spectrum, the wave operators $W_-(H,H_0)$ and $W_+(H^*,H_0)$ in dissipative scattering theory are defined in the same way as in unitary scattering theory, namely
\begin{equation*}
W_-(H,H_0) := \underset{t \to \infty}{\slim} \, e^{-\mathrm{i}tH} e^{\mathrm{i}tH_0} , \quad W_+(H^*,H_0) := \underset{t \to \infty}{\slim} \, e^{\mathrm{i}tH^*} e^{-\mathrm{i}tH_0} ,
\end{equation*}
where, recall, $\slim$ stands for strong limit.

The existence and basic properties of $W_-( H , H_0 )$ and $W_+( H^* , H_0 )$ are stated in the following proposition.
\begin{proposition}\label{prop:wave1}
Suppose that Hypotheses \ref{hyp:H0}--\ref{hyp:H3} hold. Then $W_-( H , H_0 )$ and $W_+( H^* , H_0 )$ exist and are injective contractions. Moreover,
\begin{equation*}
e^{ - \mathrm{i} t H } W_-( H , H_0 ) = W_- ( H , H_0 ) e^{ - \mathrm{i} t H_0 } , \quad e^{ - \mathrm{i} t H^* } W_+( H^* , H_0 ) = W_+ ( H^* , H_0 ) e^{ - \mathrm{i} t H_0 } ,
\end{equation*}
for all $t \in \mathbb{R}$, and
\begin{align*}
& \overline{ \mathrm{Ran} \, W_-( H , H_0 ) } = \big ( \mathcal{H}_{\mathrm{b}}(H) \oplus \mathcal{H}_{\mathrm{d}}(H^*) \big )^\perp \subseteq \mathcal{H}_{\mathrm{ac}}(H) , \\
& \overline{ \mathrm{Ran} \, W_+( H^* , H_0 ) } = \big ( \mathcal{H}_{\mathrm{b}}(H) \oplus \mathcal{H}_{\mathrm{d}}(H) \big )^\perp \subseteq \mathcal{H}_{\mathrm{ac}}(H) .
\end{align*}
\end{proposition}
\begin{proof}
See \cite[Propositions 3.4 and 3.5]{FaFr18_01}.
\end{proof}

\subsubsection{The wave operators $W_+( H_0 , H )$ and $W_-( H_0 , H^* )$}

Recall that $\mathcal{H}_{\mathrm{ac}}(H)$ and $\mathcal{H}_{\mathrm{ac}}(H^*)$ are defined in Definition \ref{def:Hac}. We denote by $\Pi_{\mathrm{ac}}(H)$, respectively $\Pi_{\mathrm{ac}}(H^*)$, the orthogonal projection onto the absolutely continuous spectral subspace of $H$, respectively $H^*$. The wave operators $W_+(H_0,H)$ and $W_-(H_0,H^*)$ are defined by
\begin{equation*}
W_+(H_0,H) := \underset{t \to \infty}{\slim} \, e^{\mathrm{i}tH_0} e^{-\mathrm{i}tH} \Pi_{ \mathrm{ac} }( H ), \quad W_-(H_0,H^*) := \underset{t \to \infty}{\slim} \, e^{-\mathrm{i}tH_0} e^{\mathrm{i}tH^*}  \Pi_{ \mathrm{ac} }( H^* ) .
\end{equation*}
Using unitarity of $e^{\mathrm{i}tH_0}$, we see that the existence of $W_+(H_0,H)$ is equivalent to the following property (sometimes called \emph{weak asymptotic completeness}): for all $u_0 \in \mathcal{H}_{\mathrm{b}}(H)^\perp = \mathcal{H}_{\mathrm{ac}}(H)$, there exists $u_+ \in \mathcal{H}$ such that $\| e^{-\mathrm{i} t H } u_0 - e^{-\mathrm{i} t H_0 } u_+ \| \to 0$ as $t \to \infty$, and in this case we have that $u_+ = W_+(H_0,H) u_0$.

The existence and basic properties of $W_+( H_0 , H )$ and $W_-( H_0 , H^* )$ are stated in the following proposition.
\begin{proposition}\label{prop:wave2}
Suppose that Hypotheses \ref{hyp:H0}--\ref{hyp:H3} hold. Then $W_+( H_0 , H )$ and $W_-( H_0 , H^* )$ exist and are contractions. Moreover,
\begin{equation*}
W_+( H_0 , H )^* = W_+ ( H^* , H_0 ) , \quad  W_-( H_0 , H^* )^* = W_- ( H , H_0 ) .
\end{equation*}
In particular,
\begin{equation*}
\mathrm{Ker} \, W_+( H_0 , H ) = \big ( \mathcal{H}_{\mathrm{b}}(H) \oplus \mathcal{H}_{\mathrm{d}}(H) \big )^\perp , \quad \mathrm{Ker} \, W_-( H_0 , H^* ) = \big ( \mathcal{H}_{\mathrm{b}}(H) \oplus \mathcal{H}_{\mathrm{d}}(H^*) \big )^\perp ,
\end{equation*}
and $W_+( H_0 , H )$ and $W_-( H_0 , H^* )$ have dense ranges.
\end{proposition}
\begin{proof}
See \cite[Proposition 3.6]{FaFr18_01}.
\end{proof}
We mention that similar results can be obtained using the Kato-Birman theory of trace-class perturbations instead of relatively smooth perturbations, see \cite{Da78_01}.

\subsubsection{The scattering operators}

In dissipative scattering theory, the scattering operators are defined by
\begin{equation*}
S( H , H_0 ) := [W_+( H^* , H_0 )]^* W_-( H , H_0 ) , \quad S( H^* , H_0 ) := [W_-( H^* , H_0 )]^* W_+( H^* , H_0 ).
\end{equation*}
These definitions generalize the usual definition of unitary scattering operators in the sense that, if $H$ is self-adjoint, then $H^*=H$ and the previous equalities reduce to the usual definitions.

Combining Propositions \ref{prop:wave1} and \ref{prop:wave2}, we arrive at the following result.
\begin{proposition}\label{prop:scatt1}
Suppose that Hypotheses \ref{hyp:H0}--\ref{hyp:H3} hold. Then $S(H,H_0)$ and $S(H^*,H_0)$ exist and are contractions. Moreover,
\begin{equation*}
e^{ - \mathrm{i} t H_0 } S( H , H_0 ) = S( H , H_0 ) e^{ - \mathrm{i} t H_0 }, \quad e^{ - \mathrm{i} t H_0 } S( H^* , H_0 ) = S( H^* , H_0 ) e^{ - \mathrm{i} t H_0 } ,
\end{equation*}
for all $t \in \mathbb{R}$ and we have that
\begin{equation*}
S(H,H_0)^* = S(H^*,H_0).
\end{equation*}
\end{proposition}
An important question, both mathematically and physically, concerns the invertibility of the scattering operators. Regarding this question, we can state the following proposition (see the next section for more precise results).
\begin{proposition}\label{prop:scatt2}
Suppose that Hypotheses \ref{hyp:H0}--\ref{hyp:H3} hold. Then the following conditions are equivalent:
\begin{enumerate}[label=\rm{(\arabic*)},leftmargin=0.8cm,itemsep=1pt]
\item $S(H,H_0)$ and $S(H^*,H_0)$ are invertible in $\mathcal{L}( \mathcal{H} )$.
\item The range of the wave operators  $W_-( H , H_0 )$ and $W_+( H^* , H_0 )$ are given by
\begin{align*}
&  \mathrm{Ran} \, W_-( H , H_0 )  = \big ( \mathcal{H}_{\mathrm{b}}(H) \oplus \mathcal{H}_{\mathrm{d}}(H^*) \big )^\perp  , \quad  \mathrm{Ran} \, W_+( H^* , H_0 ) = \big ( \mathcal{H}_{\mathrm{b}}(H) \oplus \mathcal{H}_{\mathrm{d}}(H) \big )^\perp  .
\end{align*}
\end{enumerate}
\end{proposition}
\begin{proof}
See \cite[Proposition 3.8]{FaFr18_01}.
\end{proof}

\subsubsection{The scattering matrices}

We recall that the multiplicity of the spectrum of a self-adjoint operator is defined \emph{via} the spectral theorem (see, e.g., \cite[Section VII]{ReSi80_01}). To study the scattering matrices, it is convenient to add the following condition to Hypothesis \ref{hyp:H1}(i).
\begin{assumption}[Multiplicity of $\sigma( H_0 )$]\label{hyp:H1bis}
The spectrum of $H_0$ has a constant multiplicity (which may be infinite).
\end{assumption}
To simplify notations below, we set
\begin{equation*}
\Lambda := \sigma(H_0).
\end{equation*}
Assuming Hypotheses \ref{hyp:H1}(i) and \ref{hyp:H1bis}, the spectral theorem ensures that there exists a unitary mapping from $\mathcal{H}$ to a direct integral of Hilbert spaces,
\begin{equation*}
\mathcal{F}_0 : \mathcal{H} \rightarrow  \int^\oplus_{ \Lambda } \mathcal{H}( \lambda ) \mathrm{d} \lambda ,
\end{equation*}
such that $\mathcal{F}_0  H_0 \mathcal{F}_0 ^*$ acts as multiplication by $\lambda$ on each Hilbert space $\mathcal{H}( \lambda )$. Moreover, since $\sigma( H_0 )$ has a constant multiplicity, say $k \in \mathbb{N} \cup \{ + \infty \}$, all spaces $\mathcal{H}( \lambda )$ can be identified with a fixed Hilbert space $\mathcal{M}$. Hence $\mathcal{F}_0$ becomes an operator
\begin{equation*}
\mathcal{F}_0 : \mathcal{H} \rightarrow \int^\oplus_{ \Lambda } \mathcal{M} \, \mathrm{d} \lambda = L^2 (\Lambda ;\mathcal{M}),
\end{equation*}
where ${\rm{dim}} \, \mathcal{M} =k$ and $L^2 (\Lambda ;\mathcal{M})$ is  the space of square integrable functions from $\Lambda$ to $\mathcal{M}$, (see e.g. \cite[Chapter 0, Section 1.3]{Ya10_01}). Note that in the case where $\mathcal{H} = L^2( \mathbb{R}^3 )$ and $H_0 = - \Delta$, the Hilbert space $\mathcal{M}$ is given by $\mathcal{M} = L^2 (S^2)$, where $S^2$ stands for the unit-sphere in $\mathbb{R}^3$.

Using that the scattering operator $S(H,H_0)$ commutes with $H_0$, by Proposition \ref{prop:scatt1}, one can verify that $S(H,H_0)$ admits a fiber decomposition of the form
\begin{equation*}
\mathcal{F}_0 S(H,H_0) \mathcal{F}_0^* = \int^\oplus_\Lambda S(\lambda) \mathrm{d}\lambda .
\end{equation*}
The bounded operators
\begin{equation*}
S(\lambda) \in \mathcal{L}( \mathcal{M} ),
\end{equation*}
defined for a.e. $\lambda \in \Lambda$, are called the scattering matrices (for the pair $(H,H_0)$).

One can define in the same way the scattering matrices $S^*(\lambda)$ for the pair $(H^*,H_0)$ by the relation
\begin{equation*}
\mathcal{F}_0 S(H^*,H_0) \mathcal{F}_0^* = \int^\oplus_\Lambda S^*(\lambda) \mathrm{d}\lambda .
\end{equation*}
Under the conditions of Proposition \ref{prop:scatt2}, we then have that
\begin{equation*}
[ S ( \lambda ) ]^* = S^*( \lambda ) ,
\end{equation*}
for a.e. $\lambda \in \Lambda$.

We set
\begin{equation*}
\mathcal{F}_\pm := \mathcal{F}_0 W^*_\pm( H_V , H_0 )  : \mathcal{H} \rightarrow L^2 (\Lambda ;\mathcal{M}) .
\end{equation*}
Given $s \ge 0$, an interval $X$ and a Hilbert space $\mathcal{H}$, we denote by $\mathrm{C}^{s}(X; \mathcal{H})$ the set of H{\"o}lder continuous $\mathcal{H}$-valued functions on $X$ of order $s$. In order to insure that the map $\lambda \mapsto S( \lambda )$ is continuous, it is convenient to require that the operators $V$ and $C$ are strongly smooth with respect to $H_0$ and $H_V$, respectively, in the following sense.
\begin{assumption}[Strong smoothness of $V$ with respect to $H_0$]\label{hyp:H4} $ $
\begin{enumerate}[label=\rm{(\roman*)},leftmargin=0.8cm,itemsep=1pt]
\item There exist an auxiliary Hilbert space $\mathcal{G}$ and operators $G : \mathcal{H} \to \mathcal{G}$ and $K : \mathcal{G} \to \mathcal{G}$ such that
$
V = G^* K G ,
$
with $G ( H_0^{1/2}+1 )^{-1} \in \mathcal{L}( \mathcal{H} ; \mathcal{G} )$ and $K \in \mathcal{L}( \mathcal{G} )$.
\item For all $z \in \mathbb{C}$, $\mathrm{Im}(z) \neq 0$,
$
G R_0(z) G^* \text{ is compact}.
$
\item The operator $G$ is strongly $H_0$-smooth with exponent $s_0 \in (\frac12,1)$ on any compact set $X \Subset \Lambda$, \textit{i.e.}
\begin{equation*}
\mathcal{F}_0 [ G \mathds{1}_X(H_0) ]^* : \mathcal{G} \rightarrow \mathrm{C}^{s_0}(X; \mathcal{M}) \text{ is continuous}.
\end{equation*}
\end{enumerate}
\end{assumption}
\begin{assumption}[Strong smoothness of $C$ with respect to $H_V$]\label{hyp:H5} $ $
\begin{enumerate}[label=\rm{(\roman*)},leftmargin=0.8cm,itemsep=1pt]
\item For all $z\in\mathbb{C}$,  $\mathrm{Im}(z)\neq0$, $C R_V(z) C^*$ is compact.
\item The operator $C$ is strongly $H_V$-smooth with exponent $s \in ( 0 , 1)$ on any compact set $X \Subset \Lambda$, \textit{i.e.}
\begin{equation*}
\mathcal{F}_{\pm}  [ C \mathds{1}_X(H_V) ]^* : \mathcal{H} \to \mathrm{C}^s( \Lambda ; \mathcal{H}) \text{ is continuous} .
\end{equation*}
\item The map
\begin{equation*}
\mathring{\Lambda} \in \lambda \mapsto C \big ( R_V( \lambda + i 0 ) - R_V( \lambda - i 0 ) \big ) C^* \in \mathcal{L}( \mathcal{H} ) , 
\end{equation*}
is bounded.
\end{enumerate}
\end{assumption}
We refer to \cite{Ya92_01,Ya10_01} for details on the theory of strongly smooth operators. 

In the statement below, $S^\sharp$ stands for $S$ or $S^*$. Based on a generalization of Kuroda's representation formula, the following result was established in \cite{FaNi18_01}.
\begin{proposition}
Suppose that Hypotheses \ref{hyp:H0}--\ref{hyp:H5} hold. Then, for all $\lambda \in \mathring{\Lambda}$, $S^\sharp(\lambda)$ is a contraction and $S^\sharp( \lambda ) - \mathrm{Id}$ is compact. If, in addition, $\mathrm{dim} \, \mathcal{M} = + \infty$, then for all $\lambda \in \mathring{\Lambda}$, $\| S^\sharp( \lambda ) \| = 1$ and, in particular, $\| S( H , H_0 ) \| = 1 = \| S( H^* , H_0 ) \|$.
\end{proposition}
\begin{proof}
See \cite[Theorem 2.6 and Remark 2.7]{FaNi18_01}.
\end{proof}


\section{Spectral singularities and asymptotic completeness}\label{section:3}

Our next concern is to study more precisely the invertibility of the scattering matrices and operator. Invertibility of $S(\lambda)$ is a strongly relevant physical property since it shows that to any incoming state at energy $\lambda$ corresponds a unique outgoing state and vice versa. In Section \ref{subsec:spectral}, we explain that non-invertibility of $S(\lambda)$ is equivalent to the presence of a \emph{spectral singularity} at energy $\lambda$. Section \ref{subsec:asymptotic} is devoted to the property of \emph{asymptotic completeness} of the wave operators.

\subsection{Spectral singularities}\label{subsec:spectral}

Recall that, under our assumptions and notations, the essential spectrum of $H$ is given by $\sigma_{\mathrm{ess}}(H) = \sigma( H_0 ) = \Lambda$. We recall the notion of a spectral singularity introduced in \cite{FaFr18_01,FaNi18_01}, distinguishing points in the interior of $\Lambda$ and points in the boundary $\Lambda \setminus \mathring{\Lambda}$.
\begin{definition}[Regular spectral point and spectral singularity]\label{def:spec-sing}$ $
\begin{enumerate}[label=\rm{(\roman*)},leftmargin=0.8cm,itemsep=1pt]
\item Let $\lambda \in \mathring{\Lambda}$. We say that $\lambda$ is a regular spectral point of $H$ if there exists a compact interval $K_\lambda \subset \mathbb{R}$ whose interior contains $\lambda$, such that $K_\lambda$ does not contain any accumulation point of eigenvalues of $H$, and such that the limits
\begin{equation*}
C R ( \mu - \mathrm{i} 0 ) C^* := \lim_{ \varepsilon \downarrow 0 } C R ( \mu - \mathrm{i} \varepsilon ) C^*
\end{equation*}
exist uniformly in $\mu \in K_\lambda$ in the norm topology of $\mathcal{L}( \mathcal{H} )$. If $\lambda$ is not a regular spectral point of $H$, we say that $\lambda$ is a spectral singularity of $H$.
\item Let $\lambda \in \Lambda \setminus \mathring{\Lambda}$. We say that $\lambda$ is a regular spectral point of $H$ if there exists a compact interval $K_\lambda \subset \mathbb{R}$ whose interior contains $\lambda$, such that all $\mu \in K_\lambda \cap \mathring{\Lambda}$ are regular in the sense of \emph{(i)} and such that the map
\begin{equation*}
K_\lambda \cap \mathring{\Lambda} \ni \mu \mapsto C R ( \mu - \mathrm{i} 0 ) C^* \in \mathcal{L}(\mathcal{H})
\end{equation*}
is bounded.
\item If $\Lambda$ is right-unbounded, we say that $+\infty$ is regular if there exists $m>0$ such that all $\mu \in [m,\infty) \cap \mathring{\Lambda}$ are regular in the sense of \emph{(i)} and such that the map
\begin{equation*}
[m,\infty) \cap \mathring{\Lambda} \ni \mu \mapsto C R ( \mu - \mathrm{i} 0 ) C^* \in \mathcal{L}(\mathcal{H})
\end{equation*}
is bounded.
\end{enumerate}
\end{definition}
Note that our definition of a regular spectral point is local. One can rephrase this definition saying that $\lambda$ is a regular spectral point of $H$ if the limiting absorption principle for $H$ holds in a neighborhood of $\lambda$, for the weighted resolvent $CR(z)C^*$, for values of the spectral parameter $z$ in the lower half-plane. It should be noted that we do not need to require the limiting absorption principle to hold for values of the spectral parameter in the upper half-plane: This is due to the fact that $H$ is supposed to be dissipative. We also mention that there is a natural definition of a spectral singularity for the adjoint operator $H^*$, such that $\lambda$ is a spectral singularity of $H$ if and only if $\lambda$ is a spectral singularity of $H^*$.

In the case where $H = - \Delta + V - \mathrm{i} C^*C$ on $L^2( \mathbb{R}^3 )$, with $V$ and $C$ bounded and compactly supported potentials, a spectral singularity of $H$ corresponds to a resonance embedded in the essential spectrum $[0,\infty)$ (see, e.g., \cite{DyZw19_01} for the theory of resonances for Schrödinger operators, and \cite{FaFr18_01} for a comparison between the notions of resonances and spectral singularities).

The next theorem provides several characterizations of a spectral singularity $\lambda \in \mathring{\Lambda}$. It is based, in particular, on a generalization of Kuroda's representation formula to the context of dissipative scattering theory.
\begin{theorem}\label{thm:equiv}
Suppose that Hypotheses \ref{hyp:H0}--\ref{hyp:H5} hold. Let $\lambda \in \mathring{\Lambda}$. Then the following conditions are equivalent:
\begin{enumerate}[label=\rm{(\arabic*)},leftmargin=0.8cm,itemsep=1pt]
\item $\lambda$ is a regular spectral point of $H$.
\item $\lambda$ is not an accumulation point of eigenvalues of $H$ located in $\lambda - \mathrm{i} ( 0 , \infty )$ and the limit 
\begin{equation*}
C R ( \mu - \mathrm{i} 0 ) C^* = \lim_{ \varepsilon \downarrow 0 } C R ( \mu - \mathrm{i} \varepsilon ) C^*
\end{equation*}
exists in the norm topology of $\mathcal{L}( \mathcal{H} )$.
\item The operator $\mathrm{Id} - \mathrm{i} C R_V( \lambda - \mathrm{i} 0 ) C^*$ is invertible in $\mathcal{L}( \mathcal{H} )$.
\item The scattering matrix $S(\lambda)$ is invertible in $\mathcal{L}( \mathcal{M} )$.
\end{enumerate}
\end{theorem}
\begin{proof}
See \cite[Theorem 2.9 and Lemma 4.1]{FaNi18_01}.
\end{proof}
In general, it is a difficult problem to identify explicitly the spectral singularities of a given dissipative operator. Nevertheless, one can show that the set of spectral singularities is not too large in the following sense.
\begin{proposition}
Suppose that Hypotheses \ref{hyp:H0}--\ref{hyp:H5} hold. Then the set of spectral singularities of $H$ is a closed subset of $\Lambda$ of Lebesgue measure $0$.
\end{proposition}
\begin{proof}
See \cite[Proposition 4.2]{FaNi18_01}.
\end{proof}
In Section \ref{section:4}, for the nuclear optical model, we will show that the set of spectral singularities is generically empty.

Recall from Propositions \ref{prop:wave1} and \ref{prop:scatt2} that the scattering operators $S(H,H_0)$ and $S(H^*,H_0)$ are invertible if and only if the wave operators $W_-(H,H_0)$ and $W_+(H,H_0)$ have closed ranges. The following proposition shows that the study of spectral singularities is also relevant in order to answer the question of the invertibility of the scattering operators.
\begin{proposition}
Suppose that Hypotheses \ref{hyp:H0}--\ref{hyp:H5} hold. Suppose in addition that $\Lambda \setminus \mathring{\Lambda}$ is finite and that all $\lambda \in \Lambda \setminus \mathring{\Lambda}$ are regular in the sense of Definition \ref{def:spec-sing} (if $\Lambda$ is right-unbounded, we also assume that $+\infty$ is regular). Then the following conditions are equivalent:
\begin{enumerate}[label=\rm{(\arabic*)},leftmargin=0.8cm,itemsep=1pt]
\item $S(H,H_0)$ is invertible in $\mathcal{L}( \mathcal{H} )$,
\item $S(H^*,H_0)$ is invertible in $\mathcal{L}( \mathcal{H} )$,
\item $H$ has no spectral singularities in $\mathring{\Lambda}$.
\end{enumerate}
\end{proposition}
\begin{proof}
See \cite[Theorem 2.10]{FaNi18_01}.
\end{proof}
To conclude this section, we propose the following definition of the ``order'' of a spectral singularity of $H$. It will be relevant in the next section.
\begin{definition}[Order of a spectral singularity]\label{def:order-spec-sing}
We say that $\lambda \in \mathring{\Lambda}$ is a spectral singularity of $H$ of finite order if $\lambda$ is a spectral singularity of $H$ and there exist $\nu \in \mathbb{N}^*$ and a compact interval $K_\lambda$, whose interior contains $\lambda$, such that the limits
\begin{equation*}
 \lim_{ \varepsilon \downarrow 0 } ( \mu - \lambda )^\nu C R ( \mu - \mathrm{i} \varepsilon ) C^*
\end{equation*}
exist uniformly in $\mu \in K_\lambda$ in the norm topology of $\mathcal{L}( \mathcal{H} )$. The order $\nu_0$ of the spectral singularity $\lambda$ is then defined as the minimum of all $\nu \in \mathbb{N}^*$ such that the previous limit exists.
\end{definition}
As mentioned above, if one considers the nuclear optical model $H = - \Delta + V - \mathrm{i} C^* C$ with bounded and compactly supported potentials $V$ and $C$, then a spectral singularity corresponds to a resonance in the usual sense (see, e.g., \cite{DyZw19_01}). One can then verify that the order of a spectral singularity in the sense of Definition \ref{def:order-spec-sing} corresponds to the multiplicity of the corresponding resonance, see \cite[Section 6]{FaFr18_01}.

\subsection{Asymptotic completeness}\label{subsec:asymptotic}

We are interested in this section in the property of asymptotic completeness of the wave operators. In our context, this property can be defined as follows.
\begin{definition}[Asymptotic completeness]
The wave operators $W_-( H , H_0 )$ and $W_+( H^* , H_0 )$ are said to be asymptotically complete if their ranges coincide with the orthogonal complements of all generalized eigenstates of $H$ and $H^*$, respectively. In other words,
\begin{equation*}
 \mathrm{Ran}( W_-( H , H_0 ) ) = \big ( \mathcal{H}_{ \mathrm{b} }( H ) \oplus \mathcal{H}_{ \mathrm{p} }( H^* ) \big )^\perp, \quad \mathrm{Ran}( W_+( H^* , H_0 ) ) = \big ( \mathcal{H}_{ \mathrm{b} }( H ) \oplus \mathcal{H}_{ \mathrm{p} }( H ) \big )^\perp .
\end{equation*}
\end{definition}
With the alternative definition $\tilde{\mathcal{H}}_{\mathrm{ac}}(H)$ of the absolutely continuous spectral subspace of $H$ suggested at the end of Section \ref{subsec:spectrum}, we see that the asymptotic completeness of the wave operators is the statement that $\mathrm{Ran}( W_+( H , H_0 ) ) = \tilde{\mathcal{H}}_{\mathrm{ac}}(H)$ and $\mathrm{Ran}( W_+( H^* , H_0 ) ) = \tilde{\mathcal{H}}_{\mathrm{ac}}(H^*)$.

In \cite{FaFr18_01}, asymptotic completeness is proven under the following further assumption.
\begin{assumption}[Finiteness of the number of discrete eigenvalues and spectral singularities]\label{hyp:H6}$ $
\begin{enumerate}[label=\rm{(\roman*)},leftmargin=0.8cm,itemsep=1pt]
\item $H$ has at most finitely many (discrete) eigenvalues.
\item $H$ has at most finitely many spectral singularities in $\mathring{\Lambda}$ and each spectral singularity is of finite order.
\item $\Lambda \setminus \mathring{\Lambda}$ is finite and all $\lambda \in \Lambda \setminus \mathring{\Lambda}$ are regular. Moreover, if $\Lambda$ is right-unbounded, then $+\infty$ is regular.
\end{enumerate}
\end{assumption}
We then have the following result.
\begin{theorem}
Suppose that Hypotheses \ref{hyp:H0}--\ref{hyp:H6} hold. Then
\begin{equation*}
\mathcal{H}_{ \mathrm{p} }( H ) = \mathcal{H}_{ \mathrm{d} }( H ), \quad \mathcal{H}_{ \mathrm{p} }( H^* ) = \mathcal{H}_{ \mathrm{d} }( H^* ).
\end{equation*}
Moreover, 
\begin{align*}
& W_-( H , H_0 ) \text{ and } W_+( H^* , H_0 ) \text{ are asymptotically complete} \\
&\Longleftrightarrow \quad H \text{ has no spectral singularities in } \mathring{\Lambda}.
\end{align*}
If these equivalent conditions are satisfied, then
\begin{enumerate}[label=\rm{(\arabic*)},leftmargin=0.8cm,itemsep=1pt]
\item There is an $H$-invariant direct sum decomposition
\begin{equation*}
\mathcal{H} = \big \{ \mathcal{H}_{ \mathrm{b} }( H ) \oplus \mathcal{H}_{ \mathrm{p} }( H ) \big \} \oplus \big ( \mathcal{H}_{ \mathrm{b} }( H ) \oplus \mathcal{H}_{ \mathrm{p} }( H^* ) \big )^\perp ,
\end{equation*}
and the restriction of $H$ to $\big ( \mathcal{H}_{ \mathrm{b} }( H ) \oplus \mathcal{H}_{ \mathrm{p} }( H^* ) \big )^\perp$ is similar to $H_0$. An analogous statement holds for $H^*$.
\item The wave operators $W_+( H_0 , H )$ and $W_-( H_0, H^* )$ are surjective and their kernels are given by
\begin{equation*}
\mathrm{Ker} \, W_+( H_0 , H ) = \big ( \mathcal{H}_{\mathrm{b}}(H) \oplus \mathcal{H}_{\mathrm{p}}(H) \big )^\perp , \quad \mathrm{Ker} \, W_-( H_0 , H^* ) = \big ( \mathcal{H}_{\mathrm{b}}(H) \oplus \mathcal{H}_{\mathrm{p}}(H^*) \big )^\perp .
\end{equation*}
\item The scattering operators $S(H,H_0)$ and $S(H^*,H_0)$ are bijective.
\end{enumerate}
\end{theorem}
\begin{proof}
See \cite{FaFr18_01,FaNi18_01}.
\end{proof}

\subsection{Application to the nuclear optical model}\label{sec:nuclear}

Now, we describe the main consequences of the abstract results previously stated for the \emph{nuclear optical model}. This model was introduced in \cite{FePoWe54_01} as a phenomenological model describing the possible absorption and elastic scattering of a neutron -- or a proton -- at a nucleus. In this context, the pseudo-Hamiltonian $H$ considered previously is given by a dissipative Schrödinger operator. See \cite{Ho71_01,Fe92_01} for a thorough exposition of various versions of the model and their physical interpretations, and \cite{DiCh19_01} for more recent developments.

Hence, in this section, we focus on the nuclear optical model, setting
\begin{equation*}
H_0 = - \Delta, \quad H_V = - \Delta + V(x) , \quad H = H_V - \mathrm{i} W(x) ,
\end{equation*}
on $L^2( \mathbb{R}^3 )$. We recall that the unit-sphere in $\mathbb{R}^3$ is denoted by $S^2$.  We refer to \cite{FaFr18_01,FaNi18_01} for details showing that the abstract Hypotheses \ref{hyp:H0}--\ref{hyp:H6} are indeed satisfied in the case of the nuclear optical model, under the conditions on the potentials imposed in the following theorems.
\begin{theorem}
Suppose that
\begin{enumerate}[label=\rm{(\roman*)},leftmargin=0.8cm,itemsep=1pt]
\item $V$ is real-valued, $V \in \mathrm{C}^2( \mathbb{R}^3 )$ and there exists $\rho > 3$ such that, for all $|\alpha| \le 2$, $\partial^\alpha V ( x ) = \mathcal{O}( \langle x \rangle^{-\rho-|\alpha|} )$, $|x| \to \infty$,
\item $W$ is non-negative, $W(x)>0$ on a non-trivial open set and there exists $\delta > 2$ such that $W( x ) = \mathcal{O}( \langle x \rangle^{-\delta} )$, $|x| \to \infty$,
\item $0$ is neither an eigenvalue nor a resonance of $H_V$.
\end{enumerate}
Then, for all $\lambda > 0$,
\begin{equation*}
S( \lambda ) \text{ is invertible in } \mathcal{L}( L^2( S^2 ) ) \quad \Longleftrightarrow \quad \lambda \text{ is not a spectral singularity of } H.
\end{equation*}
Moreover,
\begin{equation*}
S( H , H_0 ) \text{ is invertible in } \mathcal{L}( L^2( \mathbb{R}^3 ) ) \quad \Longleftrightarrow \quad H \text{ has no spectral singularities in } (0 , \infty) ,
\end{equation*}
and if these conditions hold, then $\mathrm{Ran} \, W_-( H , H_0 ) = \mathcal{H}_{\mathrm{d}}(H^*)^\perp$.
\end{theorem}
\begin{proof}
See \cite{FaFr18_01,FaNi18_01}.
\end{proof}
The set of bounded and compactly supported potentials from $\mathbb{R}^3$ to $\mathbb{C}$ is denoted by $L^\infty_{\mathrm{c}}( \mathbb{R}^3 )$. If we suppose that $V$ and $W$ belong to $L^\infty_{\mathrm{c}}( \mathbb{R}^3 )$, we have in addition the following more precise results.
\begin{theorem}\label{thm:cpct_support}
Suppose that
\begin{enumerate}[label=\rm{(\roman*)},leftmargin=0.8cm,itemsep=1pt]
\item $V$ is real-valued and $V \in L^\infty_{\mathrm{c}}( \mathbb{R}^3 )$.
\item $W$ is non-negative, $W(x)>0$ on a non-trivial open set and $W \in L^\infty_{\mathrm{c}}( \mathbb{R}^3 )$.
\item $0$ is neither an eigenvalue nor a resonance of $H_V$.
\end{enumerate}
Then, $\mathcal{H}_{\mathrm{p}}(H) = \mathcal{H}_{\mathrm{d}}(H)$. Moreover,
\begin{align*}
W_-( H , H_0 ) \text{ is asymptotically complete} \quad &\Longleftrightarrow \quad \mathrm{Ran} \, W_-( H , H_0 ) = \mathcal{H}_{\mathrm{p}}(H^*)^\perp \\
&\Longleftrightarrow \quad H \text{ has no spectral singularities in } ( 0 , \infty ).
\end{align*}
If these conditions hold, then
\begin{enumerate}[label=\rm{(\arabic*)},leftmargin=0.8cm,itemsep=1pt]
\item $S( H , H_0 )$ is invertible in $\mathcal{L}( L^2 ( \mathbb{R}^3 ) )$,
\item For all $\lambda > 0$, $S(\lambda)$ is invertible in $\mathcal{L}( L^2( S^2 ) )$,
\item The restriction of $H$ to $\mathcal{H}_{\mathrm{p}}(H^*)^\perp$ is similar to $H_0$.
\end{enumerate}
\end{theorem}
\begin{proof}
See \cite{FaFr18_01,FaNi18_01}.
\end{proof}
We mention that the fact that $\mathcal{H}_{\mathrm{b}}(H) = \{ 0 \}$ in the context of the present section follows from unique continuation arguments. Moreover, it is proven in \cite{Wa11_01} that $0$ cannot be a spectral singularity of $H$. On the other hand, for any $\lambda > 0$, one can construct smooth and compactly supported potentials $V$ and $W$ such that $\lambda$ is a spectral singularity of $H$ (see \cite{Wa12_01}).

\section{Generic nature of Asymptotic Completeness}\label{section:4}

In this section, our purpose is to establish that, under suitable assumptions, the wave operators $W_-(H,H_0)$ and $W_+(H^*,H_0)$ are \emph{generically} asymptotically complete. We will work in the context of the nuclear optical model of Section \ref{sec:nuclear}, where
\begin{equation*}
H = H_0 + V - \mathrm{i} W = H_V - \mathrm{i} W ,
\end{equation*}
on $\mathcal{H} = L^2( \mathbb{R}^3 )$. Here $H_0 = - \Delta$, and the real-valued potentials $V,W$ are supposed to be bounded and compactly supported, with $W \ge 0$. We set $C = \sqrt{W}$ (so that, in particular, $C^*=C$).

We will say that a property $\mathcal{P}_C$ depending on the choice of the operator $C$ is generically true if the set of $C$'s such that $\mathcal{P}_C$ holds is a countable intersection of dense open sets in a suitable Banach space. Part of our strategy will be adapted from \cite{AgHeSk89_01}.
\begin{theorem}\label{thm:generic}
Let $V \in L^\infty_{\mathrm{c}}( \mathbb{R}^3 ; \mathbb{R} )$ be such that $0$ is not an eigenvalue nor a resonance of $H_V$. Then, for all $\Omega \subset \mathbb{R}^3$ compact, the set
\begin{align*}
& \big \{ C \in L^\infty ( \mathbb{R}^3 ; \mathbb{R} ) , \, \mathrm{supp}( C ) \subset \Omega \text{ and } H_V - \mathrm{i} C^*C \text{ has no spectral singularities in } (0 , \infty ) \big \} \\
&= \big \{ C \in L^\infty ( \mathbb{R}^3 ; \mathbb{R} ) , \, \mathrm{supp}( C ) \subset \Omega \text{ and } W_-( H_V - \mathrm{i} C^*C , H_0 ) \text{ is asymptotically complete} \big \} 
\end{align*}
is a countable intersection of dense open sets in $\{ C \in L^\infty ( \mathbb{R}^3 ; \mathbb{R} ) , \, \mathrm{supp}( C ) \subset \Omega \}$, for the topology induced by the $\|\cdot\|_\infty$-norm.
\end{theorem}
Note that the equality in the statement of Theorem \ref{thm:generic} is a consequence of Theorem \ref{thm:cpct_support}. For all compact interval $J \subset ( 0 , \infty )$, we set
\begin{align*}
\mathcal{E}_J := \big \{ C \in L^\infty ( \mathbb{R}^3 ; \mathbb{R} ) , \, \mathrm{supp}( C ) \subset \Omega \text{ and } H_V - \mathrm{i} C^*C \text{ has no spectral singularities in } J \big \}.
\end{align*}
To establish Theorem \ref{thm:generic}, it then suffices to show that, for all compact interval $J \subset ( 0 , \infty )$, $\mathcal{E}_J$ is open and dense in $\{ C \in L^\infty ( \mathbb{R}^3 ; \mathbb{R} ) , \, \mathrm{supp}( C ) \subset \Omega \}$. This is the purpose of the following two lemmas.

In order to underline the dependence on $C$ of the pseudo-Hamiltonian $H$, we will use in this section the notation
\begin{equation*}
H_{V,C} := H_V - \mathrm{i} C^* C = H_V - \mathrm{i} C^2 ,
\end{equation*}
for all $C \in L^\infty( \mathbb{R}^3 ; \mathbb{R} )$. We also set
\begin{equation*}
L^\infty_\Omega := \big \{ C \in L^\infty ( \mathbb{R}^3 ; \mathbb{R} ) , \, \mathrm{supp}( C ) \subset \Omega \big \} ,
\end{equation*}
for all compact set $\Omega \subset \mathbb{R}^3$.
\begin{lemma}\label{lm:open}
Let $V \in L^\infty_{\mathrm{c}}( \mathbb{R}^3 ; \mathbb{R} )$ be such that $0$ is not an eigenvalue nor a resonance of $H_V$ and let $\Omega \subset \mathbb{R}^3$ be a compact set. Assume that there are a compact interval $J \subset ( 0 , \infty )$ and $C_0 \in L^\infty( \mathbb{R}^3 ; \mathbb{R} )$, $\mathrm{supp}( C_0 ) \subset \Omega$, such that $H_{V,C_0}$ has no spectral singularities in $J$. Then there exists $r > 0$ such that, for all $C \in L^\infty_\Omega$ satisfying $\| C - C_0 \|_\infty \le r$,
\begin{equation*}
H_{V,C} \text{ has no spectral singularities in } J.
\end{equation*}
\end{lemma}
\begin{proof}
Let $\Omega$, $J$ and $C_0$ be as in the statement of the lemma. Let $\lambda_0 \in J$. By assumption, $\lambda_0$ is a regular spectral point of $H_{V,C_0}$ and therefore, by Theorem \ref{thm:equiv}, we know that $\mathrm{Id} - \mathrm{i} C_0 R_V( \lambda_0 - \mathrm{i} 0 ) C_0^*$ is invertible in $\mathcal{L}( \mathcal{H} )$. Since 
\begin{equation*}
L^\infty_\Omega \times ( 0 , \infty ) \ni ( C , \lambda ) \mapsto C R_V( \lambda - \mathrm{i} 0 ) C^* = C \mathds{1}_\Omega R_V( \lambda - \mathrm{i} 0 ) \mathds{1}_\Omega C^* \in \mathcal{L}( \mathcal{H} )
\end{equation*}
is continuous under our assumptions, we deduce that there exist $r_0>0$ and a neighborhood $\mathcal{U}_{\lambda_0} \subset \mathbb{R}$ of $\lambda_0$ such that, for all $C \in L^\infty_\Omega$ such that $ \| C - C_0 \| \le r_0$ and all $\lambda \in \mathcal{U}_{\lambda_0}$, $\mathrm{Id} - \mathrm{i} C R_V( \lambda - \mathrm{i} 0 ) C^*$ is invertible in $\mathcal{L}( \mathcal{H} )$. Equivalently, by Theorem \ref{thm:equiv}, we have that $\lambda$ is a regular spectral point of $H_{V,C}$ for all $C \in L^\infty_\Omega$ such that $ \| C - C_0 \| \le r_0$ and all $\lambda \in \mathcal{U}_{\lambda_0}$.

Now, we have the inclusion
\begin{equation*}
J \subset \bigcup_{\lambda \in J} \mathcal{U}_\lambda ,
\end{equation*}
and since $J$ is compact, we deduce that there are $\lambda_1 , \dots \lambda_n \in J$ such that $J \subset \mathcal{U}_{\lambda_1} \cup \cdots \cup \mathcal{U}_{\lambda_n}$. Setting $r = \min( r_1 , \dots , r_n )$, we conclude that for all $C \in L^\infty_\Omega$ such that $ \| C - C_0 \| \le r_0$, $H_{V,C}$ has no spectral singularities in $J$.
\end{proof}
Lemma \ref{lm:open} shows that, given a compact interval $J \subset (0,\infty)$, the set $\mathcal{E}_J$ is open in $L^\infty_\Omega$. Our next purpose is to prove that $\mathcal{E}_J$ is dense in $L^\infty_\Omega$. We recall that for any $V,C \in L^\infty_{\mathrm{c}}( \mathbb{R}^3 )$, $H_{V,C}$ has at most finitely many spectral singularities in $(0,\infty)$ counting orders. This follows from the theory of resonances (see \cite{DyZw19_01} and \cite[Section 6]{FaFr18_01} for more details).
\begin{lemma}\label{lm:dense}
Let $V \in L^\infty_{\mathrm{c}}( \mathbb{R}^3 ; \mathbb{R} )$ be such that $0$ is not an eigenvalue nor a resonance of $H_V$ and let $\Omega \subset \mathbb{R}^3$ be a compact set. Let $J \subset ( 0 , \infty )$ be a compact interval. For all $C_0 \in L^\infty_\Omega$ and all $\varepsilon >0$, there exists $C \in L^\infty_\Omega$ such that $\| C - C_0\|_\infty\le\varepsilon$ and $H_{V,C}$ has no spectral singularities in $J$.
\end{lemma}
\begin{proof}
Let $V$, $\Omega$ and $J$ be as in the statement of the lemma. Assume by contradiction that there exist $C_0 \in L^\infty_\Omega$ and $\varepsilon_0>0$ such that, for all $C \in L^\infty_\Omega$ such that $\| C - C_0\|_\infty\le\varepsilon_0$, $H_{V,C}$ has spectral singularities in $J$.

In a first step, we use that resonances are generically simple. Namely, adapting the proof of \cite[Theorem 3.14]{DyZw19_01} in a straightforward way (see also \cite{KlZw95_01}), one can show that there exists $\tilde{C}_0 \in L^\infty_\Omega$ such that $\| \tilde{C}_0 - C_0\|_\infty\le\varepsilon_0/2$ and all spectral singularities of $H_{V,\tilde{C}_0}$ are at most of order $1$.

Now, let $\lambda_1 , \dots , \lambda_n$ be the spectral singularities (of order $1$) of $H_{V,\tilde{C}_0}$ in $J$. If $\tilde{C}_0 = 0$, then this set is empty and we obtain a contradiction. Hence we assume in the following that $\tilde{C}_0 \neq 0$. We introduce a real parameter $g$ in the pseudo-Hamiltonian, considering the family of operators 
\begin{equation*}
H_{V,g\tilde{C}_0} = H_V - \mathrm{i} g^2 \tilde{C}_0^*\tilde{C}_0 ,
\end{equation*}
for $g$ close to $1$. We claim that, for all $j \in \{ 1 , \dots , n \}$, there exist a neighborhood $\mathcal{V}_{\lambda_j} \subset \mathbb{R}$ of $\lambda_j$ and $\varepsilon_j > 0$ such that, for all $g \in \mathbb{R}$ satisfying $0 < |g-1|\le\varepsilon_j$, $H_{V,g\tilde{C}_0}$ has no spectral singularities in $\mathcal{V}_j$. Indeed, let
\begin{equation*}
A( \lambda ) := \tilde{C}_0 R_V( \lambda - \mathrm{i} 0 ) \tilde{C}_0^*.
\end{equation*}
From Theorem \ref{thm:equiv} and the fact that $A( \lambda )$ is compact, we deduce that $\lambda$ is a spectral singularity of $H_{V,g\tilde{C}_0}$ if and only if $1$ is an eigenvalue of $\mathrm{i} g^2 A( \lambda )$. In particular, if $g=1$ and $\lambda=\lambda_j$, we see that $1$ is a simple, discrete eigenvalue of $\mathrm{i} A( \lambda_j )$. Therefore, if $\Gamma_j$ is a curve oriented counterclockwise, whose interior contains $1$ and no other eigenvalue of $\mathrm{i} A( \lambda_j )$, it follows from standard perturbation theory that, for any $\lambda$ in a neighborhood of $\lambda_j$ and $g$ close to $1$, $\mathrm{i} g^2 A( \lambda )$ has a unique eigenvalue, say $\mu_{g,\lambda}$, in the interior of $\Gamma_j$. We shall show that $\mu_{g,\lambda} \neq 1$ except if $g = 1$ and $\lambda = \lambda_j$.

Clearly, we have that $\mu_{g,\lambda} = g^2 \mu_{1,\lambda}$ and $\mu_{1,\lambda_j} = 1$. Moreover, letting 
\begin{equation*}
\pi_{g,\lambda} := \frac{ 1 }{ 2 \mathrm{i} \pi } \int_{\Gamma_j} \big( z - \mathrm{i} g^2 A( \lambda ) \big )^{-1} \mathrm{d} z 
\end{equation*}
be the Riesz projection corresponding to $\mu_{g,\lambda}$, we have that $\pi_{g,\lambda} = \pi_{1,\lambda}$. 

Let $u_j \in \mathrm{Ran}( \pi_{1,\lambda_j} )$ be a normalized eigenstate of $\mathrm{i} A ( \lambda_j )$ corresponding to the eigenvalue $1$, $\mathrm{i} A ( \lambda_j ) u_j = \lambda_j u_j$, $\| u_j \| = 1$. Then, for $(g,\lambda)$ near $(1,\lambda_j)$, $\pi_{g,\lambda} u_j \neq 0$ and $\pi_{g,\lambda} u_j$ is an eigenstate of $\mathrm{i} g^2 A ( \lambda )$ corresponding to $\mu_{g,\lambda}$. We compute
\begin{align*}
\mathrm{Im} ( \mu_{g,\lambda} ) &= \| \pi_{g,\lambda} u_j \|^{-2} \mathrm{Im} \langle \pi_{g,\lambda} u_j , \mathrm{i} g^2 A( \lambda ) \pi_{g,\lambda} u_j \rangle \\
&= \frac12 \| \pi_{g,\lambda} u_j \|^{-2} g^2 \big \langle u_j , \pi_{g,\lambda}^* \big ( \tilde{C}_0 R_V( \lambda -\mathrm{i} 0 ) \tilde{C}_0^* + \tilde{C}_0 R_V( \lambda + \mathrm{i} 0 ) \tilde{C}_0^* \big ) \pi_{g,\lambda} u_j \big \rangle \\
&= \frac12 \| \pi_{g,\lambda} u_j \|^{-2} g^2 \big \langle u_j , \pi_{1,\lambda}^* \big ( \tilde{C}_0 R_V( \lambda -\mathrm{i} 0 ) \tilde{C}_0^* + \tilde{C}_0 R_V( \lambda + \mathrm{i} 0 ) \tilde{C}_0^* \big ) \pi_{1,\lambda} u_j \big \rangle ,
\end{align*}
where we used that $\pi_{g,\lambda} = \pi_{1,\lambda}$ in the last equality. Hence we see that $\mathrm{Im} ( \mu_{g,\lambda} ) = 0$ if and only if the scalar product in the previous equality vanishes. The maps $\lambda \mapsto A( \lambda )$ and $\lambda \mapsto A(\lambda)^*$ are real analytic in a neighborhood of $\lambda_j$. This implies that
\begin{equation*}
\lambda \mapsto \big \langle u_j , \pi_{1,\lambda}^* \big ( \tilde{C}_0 R_V( \lambda -\mathrm{i} 0 ) \tilde{C}_0^* + \tilde{C}_0 R_V( \lambda + \mathrm{i} 0 ) \tilde{C}_0^* \big ) \pi_{1,\lambda} u_j \big \rangle
\end{equation*}
has a unique zero in a neighborhood of $\lambda_j$. But this zero is $\lambda_j$ since $\mathrm{Im}( \mu_{g,\lambda_j} ) = \mathrm{Im} ( g^2 ) = 0$.

Hence we have proven that for all $\lambda$ in a neighborhood $\mathcal{V}_j$ of $\lambda_j$, $\lambda \neq \lambda_j$, and all $g$ in a neighborhood of $1$, $\mathrm{Im}( \mu_{g,\lambda} ) \neq 0$. In particular, $\mu_{g,\lambda} \neq 1$. It remains to show that $\mu_{g,\lambda_j} \neq 1$ except if $g = 1$. But this is obvious, since $\mu_{g,\lambda_j} = g^2$.

Summarizing, for all $j \in \{ 1 , \dots , n \}$, there exist a neighborhood $\mathcal{V}_j \subset \mathbb{R}$ of $\lambda_j$ and $\varepsilon_j > 0$ such that, for all $0 < | g - 1 | \le \varepsilon_j$, $H_{V,g\tilde{C}_0}$ has no spectral singularities in $\mathcal{V}_j$. Moreover, by assumption, $H_{V,\tilde{C}_0}$ has no spectral singularities in $J \setminus \cup_{j=1}^n \mathcal{V}_j$. By Lemma \ref{lm:open}, this implies that there exists $r > 0$ such that, for all $0 < | g - 1 | \le r$, $H_{V,g\tilde{C}_0}$ has no spectral singularities in $J \setminus \cup_{j=1}^n \mathcal{V}_j$. Picking $g$ such that $0 < |g - 1 |\le \min ( \varepsilon_1 , \dots , \varepsilon_n , r , \| \tilde{C}_0 \|^{-1}\varepsilon_0/2)$, we obtain that $H_{V,g\tilde{C}_0}$ has no spectral singularities in $J$. Since, in addition, we have that
\begin{equation*}
\big \| g \tilde{C}_0 - C_0 \big \| \le | g - 1 |\big\| \tilde{C}_0 \big \| + \big \| \tilde{C}_0 - C_0 \big \| \le \varepsilon_0 ,
\end{equation*}
this gives a contradiction. This concludes the proof of the lemma.
\end{proof}
Now, we can combine the previous two lemmas to complete the proof of Theorem \ref{thm:generic}.
\begin{proof}[Proof of Theorem \ref{thm:generic}]
Let $J \subset ( 0 , \infty )$ be a compact interval. By Lemma \ref{lm:open}, the set $\mathcal{E}_J$ is open in $L^\infty_\Omega$, while, by Lemma \ref{lm:dense}, $\mathcal{E}_J$ is dense in $L^\infty_\Omega$. The statement of the theorem then follows from the fact that
\begin{align*}
& \big \{ C \in L^\infty ( \mathbb{R}^3 ; \mathbb{R} ) , \, \mathrm{supp}( C ) \subset \Omega \text{ and } H_V - \mathrm{i} C^*C \text{ has no spectral singularities in } (0 , \infty ) \big \} \\
&= \bigcap_{ n \in \mathbb{N}^* } \mathcal{E}_{[ n^{-1} , n ]}.
\end{align*}
\end{proof}

\vspace{0.2cm}

\noindent \textbf{Acknowledgements.} I would like to thank H. Cornean, S. Fournais and J.S. M{\o}ller for the invitation to give a talk at the conference \emph{QMath14: Mathematical Results in Quantum Physics}, Aarhus, August 2019. I am grateful to E. Skibsted for a useful conversation. I warmly thank J. Fröhlich and F. Nicoleau for fruitful collaborations.

\end{document}